\documentclass[floatfix, noshowpacs, preprintnumbers, twocolumn, amsmath, amssymb, aps, prl,superscriptaddress]{revtex4-1}
\pdfoutput=1

\usepackage{bm}
\usepackage{amsmath}
\usepackage{amssymb}
\usepackage{graphicx}
\usepackage{xcolor}
\usepackage{bbold}
\usepackage{url}
\usepackage[caption=false]{subfig}

\newcommand{\bra}[1]{\left< #1 \right|} 
\newcommand{\ket}[1]{\left| #1 \right>}

\newtheorem{theorem}{Theorem}

\newenvironment{proof}[1][Proof:]{\begin{trivlist} \item[\hskip \labelsep {\bfseries #1}]}{\end{trivlist}}
\newcommand{\qed}{\nobreak \ifvmode \relax \else
      \ifdim\lastskip<1.5em \hskip-\lastskip
      \hskip1.5em plus0em minus0.5em \fi \nobreak
      \vrule height0.75em width0.5em depth0.25em\fi}

\begin{document}

\title{General rules for bosonic bunching in multimode interferometers}

\author{Nicol\`o Spagnolo}
\affiliation{Dipartimento di Fisica, Sapienza Universit\`{a} di Roma,
Piazzale Aldo Moro 5, I-00185 Roma, Italy}

\author{Chiara Vitelli}
\affiliation{Center of Life NanoScience @ La Sapienza, Istituto
Italiano di Tecnologia, Viale Regina Elena, 255, I-00185 Roma, Italy}
\affiliation{Dipartimento di Fisica, Sapienza Universit\`{a} di Roma,
Piazzale Aldo Moro 5, I-00185 Roma, Italy}

\author{Linda Sansoni}
\affiliation{Dipartimento di Fisica, Sapienza Universit\`{a} di Roma,
Piazzale Aldo Moro 5, I-00185 Roma, Italy}

\author{Enrico Maiorino}
\affiliation{Dipartimento di Fisica, Sapienza Universit\`{a} di Roma,
Piazzale Aldo Moro 5, I-00185 Roma, Italy}

\author{Paolo Mataloni}
\affiliation{Dipartimento di Fisica, Sapienza Universit\`{a} di Roma,
Piazzale Aldo Moro 5, I-00185 Roma, Italy}
\affiliation{Istituto Nazionale di Ottica (INO-CNR), Largo E. Fermi 6, I-50125 Firenze, Italy}

\author{Fabio Sciarrino}
\email{fabio.sciarrino@uniroma1.it}
\affiliation{Dipartimento di Fisica, Sapienza Universit\`{a} di Roma,
Piazzale Aldo Moro 5, I-00185 Roma, Italy}
\affiliation{Istituto Nazionale di Ottica (INO-CNR), Largo E. Fermi 6, I-50125 Firenze, Italy}

\author{Daniel J. Brod}
\affiliation{Instituto de F\'isica, Universidade Federal Fluminense, Av. Gal. Milton Tavares de Souza s/n, Niter\'oi, RJ, 24210-340, Brazil}

\author{Ernesto F. Galv\~{a}o}
\email{ernesto@if.uff.br}
\affiliation{Instituto de F\'isica, Universidade Federal Fluminense, Av. Gal. Milton Tavares de Souza s/n, Niter\'oi, RJ, 24210-340, Brazil}

\author{Andrea Crespi}
\affiliation{Istituto di Fotonica e Nanotecnologie, Consiglio
Nazionale delle Ricerche (IFN-CNR), Piazza Leonardo da Vinci, 32,
I-20133 Milano, Italy}
\affiliation{Dipartimento di Fisica, Politecnico di Milano, Piazza
Leonardo da Vinci, 32, I-20133 Milano, Italy}

\author{Roberta Ramponi}
\affiliation{Istituto di Fotonica e Nanotecnologie, Consiglio
Nazionale delle Ricerche (IFN-CNR), Piazza Leonardo da Vinci, 32,
I-20133 Milano, Italy}
\affiliation{Dipartimento di Fisica, Politecnico di Milano, Piazza
Leonardo da Vinci, 32, I-20133 Milano, Italy}

\author{Roberto Osellame}
\email{roberto.osellame@polimi.it}
\affiliation{Istituto di Fotonica e Nanotecnologie, Consiglio
Nazionale delle Ricerche (IFN-CNR), Piazza Leonardo da Vinci, 32,
I-20133 Milano, Italy}
\affiliation{Dipartimento di Fisica, Politecnico di Milano, Piazza
Leonardo da Vinci, 32, I-20133 Milano, Italy}

\begin{abstract}
We perform a comprehensive set of experiments that characterize bosonic bunching of up to 3 photons in interferometers of up to 16 modes. Our experiments verify two rules that govern bosonic bunching. The first rule, obtained recently in \cite{Aaronson10, Arkhipov11}, predicts the average behavior of the bunching probability and is known as the bosonic birthday paradox. The second rule is new, and establishes a $n!$-factor quantum enhancement for the probability that all $n$ bosons bunch in a single output mode, with respect to the case of distinguishable bosons. Besides its fundamental importance in phenomena such as Bose-Einstein condensation, bosonic bunching can be exploited in applications such as linear optical quantum computing and quantum-enhanced metrology.
\end{abstract}

\maketitle

Bosons and fermions exhibit distinctly different statistical behaviors. For fermions, the required wave-function anti-symmetrization results in the Pauli exclusion principle, exchange forces and, ultimately, in all the main features of electronic transport in solids. Bosons, on the other hand, tend to occupy the same state more often than classical particles do. This bosonic bunching behavior is responsible for fundamental phenomena such as Bose-Einstein condensation, which has been observed in a large variety of bosonic systems \cite{Klaers2010,Anders1995,Davis95}. In quantum optics, a well-known bosonic bunching effect is the Hong-Ou-Mandel two-photon coalescence \cite{HOM87} observed in balanced beam-splitters as well as in  multimode linear optical interferometers \cite{Ou06, Liu10, Peru2011, Spag2012, Crespi2012, Broome2013, Spring2013,Till2012}, and which can be generalized to a larger number of particles \cite{Ou99,Niu09,Tichy10}. In such a process, two indistinguishable photons impinging on the input ports of a balanced beam-splitter will exit from the same output port, while distinguishable photons have a non-zero probability of exiting from different ports. Besides its importance for tests on the foundations of quantum mechanics \cite{Pan12}, bosonic bunching is useful in applications such as quantum-enhanced metrology \cite{Giov2011} and photonic quantum computation \cite{Kok2007}.

\begin{figure*}[ht!]
\includegraphics[width=0.85\textwidth]{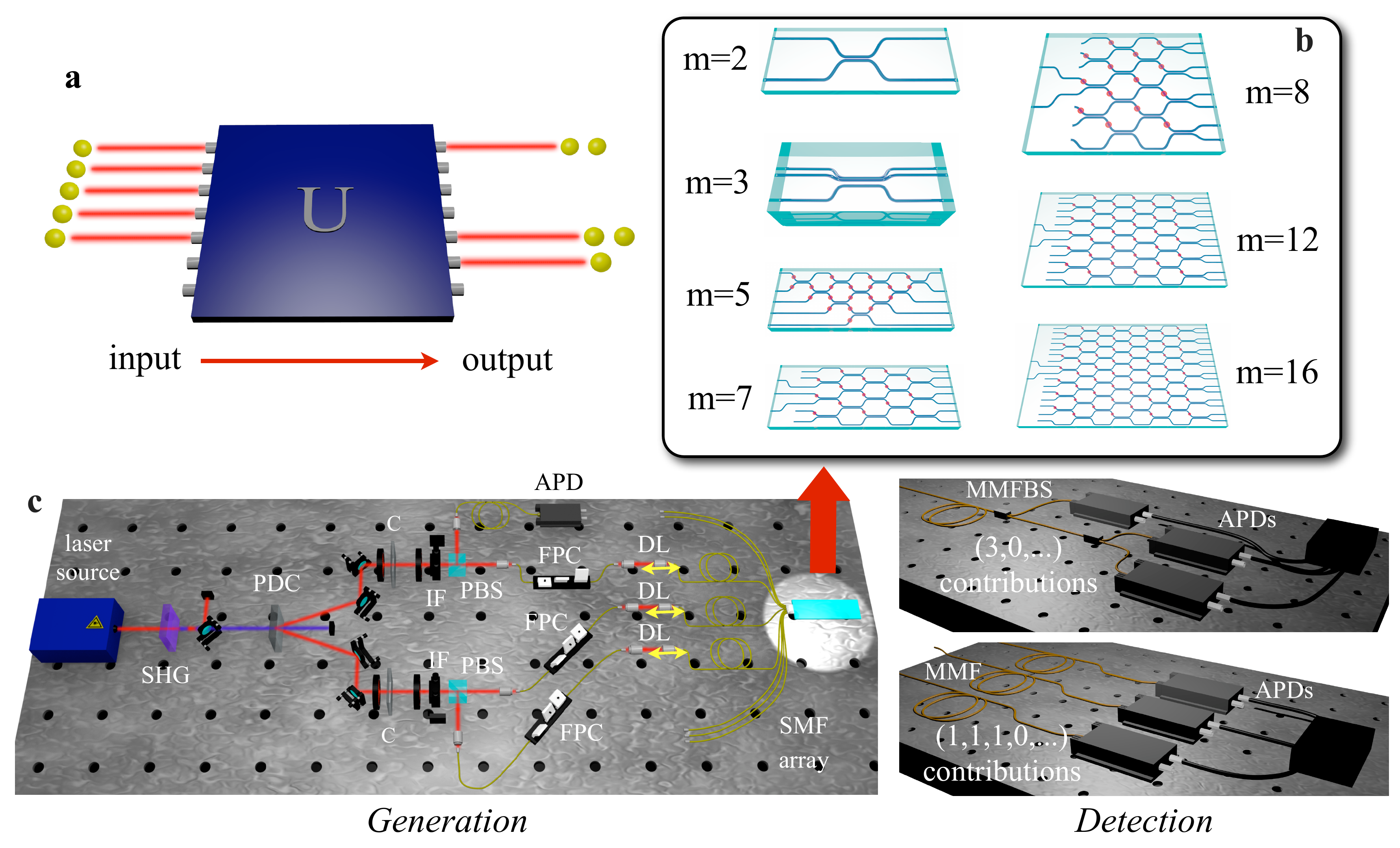}
\caption{{\bf Experimental platform for photonic bunching experiments.} {\bf a}, Input bosons evolve via a linear interferometer described by a $m \times m$ unitary $U$. A bunching event results when two (or more) bosons emerge from the same output port. {\bf b}, Architectures of the integrated linear optical interferometers exploited in the experimental verification. Red spots represent phase shifters while the directional couplers perform the beam splitter trasformation, whose reflectivity can be varied by modulating the coupling coefficient in the interaction region. {\bf c}, Experimental layout evidencing the generation and detection of single photons. (Legend - SHG: second harmonic generation, PDC: parametric down-conversion, C: walk-off compensation, IF: interference filter, PBS: polarizing beam-splitter, FPC: fiber polarization compensator, DL: delay line, SMF: single-mode fiber, MMF: multi-mode fiber, MMFBS: multi-mode fiber beam-splitter, APD: avalanche photodiode, Black-box: electronic coincidence apparatus). See \cite{SuppInf} for more experimental details.}
\label{fig:BP_exp}
\end{figure*}

In this Letter we report on the experimental verification of two general rules that govern the bosonic bunching behavior. One of them, the bosonic full-bunching rule, is theoretically proposed and proven in this Letter. The experimental verification of these rules is achieved with a comprehensive set of experiments that characterize the bosonic bunching of photons as they exit a number of multimode interferometers. Our experiments involve inputting $n$ photons (both distinguishable and indistinguishable) in different input ports of a $m$-mode linear interferometer (described by an $m \times m$ unitary $U$), and measuring the probability of each possible output distribution.

Let us now state the two bosonic bunching rules that we experimentally verify in this work. 
\begin{theorem} \textbf{(Average bosonic bunching probability \cite{Aaronson10, Arkhipov11}.)}
For an ensemble of uniformly-drawn $m$-mode random interferometers, the average probability that two or more bosons (out of the $n$ input bosons) will exit in the same output port is given by:
\begin{equation}
p_b(n,m)=1-\prod_{a=0}^{n-1}\frac{1-a/m}{1+a/m}.
\end{equation}
\end{theorem}

By uniformly-drawn, we mean unitaries picked randomly according to the unique Haar uniform distribution \cite{Zyck94} over $m \times m$ unitaries $U$; the average bunching probability pertains to the output obtained by any chosen input state evolving in this ensemble of interferometers. This rule was obtained recently in a study of a generalization of the classical birthday paradox problem to the case of bosons  \cite{Aaronson10, Arkhipov11}.

\begin{theorem} \textbf{(Full-bunching bosonic probability ratio}.)
Let $g_k$ denote the occupation number  of input mode $k$. Let us denote the probabilities that all $n$ bosons leave the interferometer in mode $j$ by $q_c(j)$ (distinguishable bosons) and  $q_q(j)$ (indistinguishable bosons). Then the ratio of full-bunching probabilities $r_{fb}=q_q(j)/q_c(j)=n!/\prod_k{g_k!}$, independently of $U,m,j$.
\end{theorem}
\begin{proof} As described in \cite{Aaronson10} for example,  the $m \times m$ unitary $U$ describing the interferometer induces a unitary $U_F$ acting on the Hilbert space of $n$ photons in $m$ modes; the probability amplitude associated with input $\ket{G}=\ket{g_1 g_2 ... g_m}$ and output $\ket{H} = \ket{h_1 h_2 ... h_m}$ is given by
\begin{equation} \label{permanent}
\bra{H} U_F\ket{G} = \frac{\textrm{per}(U_{G,H})}{\sqrt{g_1! .. g_m! h_1! .. h_m!}},
\end{equation}
where $U_{G,H}$ is the matrix obtained by repeating $g_i$ times the $i^{th}$ row of $U$, and $h_j$ times its $j^{th}$ column \cite{Scheel04}, and per$(A)$ denotes the permanent of matrix $A$ \cite{Valiant79}.

The probability that a single boson entering mode $j$ will exit in mode $i$ is $|U_{i,j}|^2$, as is easy to check \cite{Aaronson10}. Then a simple counting argument gives the probability $p_{G,H}$ that distinguishable bosons will enter the interferometer with occupation numbers $g_1 g_2 ... g_m$ and leave with occupation numbers $h_1 h_2 ... h_m$:
\begin{equation}\label{dist}
p_{G,H}=\textrm{per}(|U_{G,H}|^2)/(\prod_i{(h_i}!)),
\end{equation}
where $|U_{G,H}|^2$ is the matrix obtained by taking the absolute value squared of each corresponding element of $U_{G,H}$.

Let us now introduce an alternative, convenient way of representing the input occupation numbers. Define a $n$-tuple of $m$ integers $r_i$ so that the first $g_1$ integers are 1, followed by a sequence of $g_2$ 2's, and so on until we have $g_m$ $m$'s. As an example, input occupation numbers $g_1=2, g_2=1, g_3=0, g_4=3$ would give $r=(1,1,2,4,4,4)$. 
Using Eq. (\ref{permanent}) we can evaluate the probability $q_q(j)$ that the $n$ indistinguishable bosons will all exit in mode $j$:
\begin{equation}
q_q(j)=|\textrm{per}(A)|^2/(n!\prod_k{(g_k!)}),
\end{equation}
where $A$ is a $n \times n$ matrix with elements $A_{i,k}=U_{j,r_k}$. Since all rows of $A$ are equal, $\textrm{per}(A)$ is a sum of $n!$ identical terms, each equal to $\prod_k U_{j,r_k}$. Hence
\begin{equation}
\begin{aligned}
q_q(j)&=|n! \prod_k U_{j,r_k}|^2/(n!\prod_k{(g_k!)})= \\
&= n! |\prod_k U_{j,r_k}|^2/\prod_k{(g_k!)}.
\end{aligned}
\end{equation}

Using Eq. (\ref{dist}), we can calculate the probability $q_c(j)$ that $n$ distinguishable bosons will leave the interferometer in mode $j$: $q_c(j)=\textrm{per}(B)/n!$, where $B$ has elements $B_{i,k}=|A_{i,k}|^2=|U_{j,r_k}|^2$. Hence
\begin{equation}
q_c(j)=n! \prod_k |U_{j,r_k}|^2/n!= \prod_k |U_{j,r_k}|^2.
\end{equation}
Our new bosonic full-bunching rule establishes the value of the quantum/classical full-bunching ratio, which we can now calculate to be
\begin{equation}
r_{fb}=q_q(j)/q_c(j)=n!/\prod_k{(g_k!)}.
\end{equation} \qed
\end{proof}

Our new bosonic full-bunching rule generalizes the Hong-Ou-Mandel effect into a universal law, now applicable to any interferometer and any number of photons $n$. Despite becoming exponentially rare as $n$ increases \cite{Aaro12}, full-bunching events are enhanced by a factor as high as $n!$ when at most one boson is injected into each input mode, as in our photonic experiments.

\begin{figure*}[ht!]
\includegraphics[width=1.\textwidth]{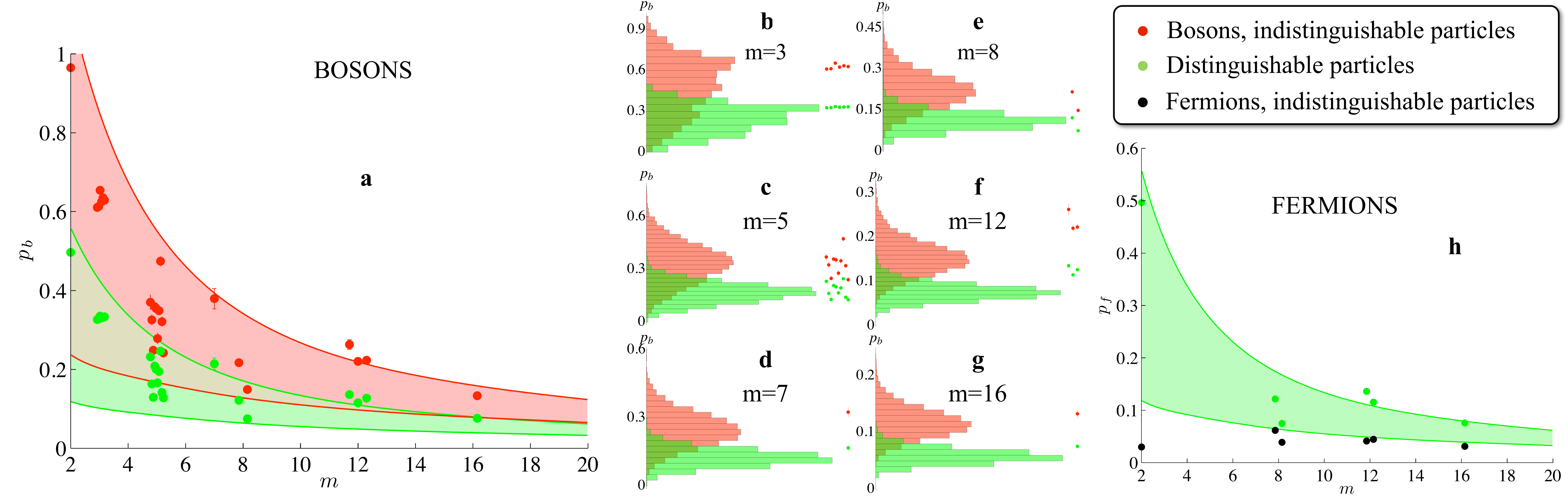}
\caption{{\textbf{Two-photon photonic bunching data.}} {\bf a}, Bunching probability $p_{b}$ as a function of the number of modes $m$ for two indistinguishable photons ($p_{b}^{(q)}$, red points) and two distinguishable photons ($p_{b}^{(c)}$, green points). We performed experiments with different unitaries ($m=3$, $m=8$, $m=12$) or different input states ($m=3$, $m=5$). Shaded areas correspond to the interval $[\overline{p_{b}} - 1.5 \, \sigma; \overline{p_{b}} + 1.5 \, \sigma]$ obtained with a numerical sampling over $10000$ uniformly random unitaries, $\overline{p_{b}}$ being the average bunching probability and $\sigma$ its standard deviation. Red area: indistinguishable photons. Green area: distinguishable photons.  {\bf b-g}, Experimental results (points) together with histograms showing the distribution of the bunching probabilities obtained with the numerical simulation. {\bf h}, Results for the bunching probability $p_{f}$ as a function of the number of modes $m$ for two indistinguishable fermions ($p_{f}^{(q)}$, black points) and two distinguishable particles ($p_{f}^{(c)}$, green points). Non-zero bunching probabilities have to be attributed to imperfections in the state preparation. Error bars in the experimental data are due to the Poissonian statistics of the measured events, and where not visible are smaller than the symbol.
}
\label{twophotons}
\end{figure*}

For our experiments, we fabricated integrated optical interferometers in a borosilicate glass by femtosecond laser waveguide writing \cite{gattass2008flm,dellavalle2009mpd}. This technique consists in a direct inscription of waveguides in the volume of the transparent substrate, exploiting the nonlinear absorption of focused femtosecond pulses to induce a permanent and localized increase in the refractive index. Single photons may jump between waveguides by evanescent coupling in regions where waveguides are brought close together; precise control of the coupling between the waveguides and of the photon path length, enabled by a 3D waveguide design \cite{Crespi2012}, provides arbitrary interferometers with different topologies (Fig. \ref{fig:BP_exp} {\bf b}).  Randomness is purposefully incorporated in our interferometer designs in various ways, for example by choosing balanced couplings together with random phase shifters, or even by decomposing a uniformly chosen unitary into arbitrary couplers and phase shifters which are then inscribed on the chip; for more details on the different architectures see \cite{SuppInf}.

Our inputs are Fock states of two or three individual photons obtained by a type-II parametric down-conversion (PDC) source (Fig. \ref{fig:BP_exp} {\bf c}). Three-photon input states result from the second-order PDC process, with the fourth photon used as a trigger. As described in \cite{Spag2012,Crespi2012,SuppInf}, the three-photon state is well-modeled as a mixture of two indistinguishable photons and a distinguishable one (probability $1- \alpha^2$), and three indistinguishable photons (probability $\alpha^2$), thus defining the indistinguishability parameter $\alpha$, which we estimated to be $\alpha=0.63 \pm 0.03$ using a standard Hong-Ou-Mandel experiment. Controllable delays between the input photons are used to change the regime from classical distinguishability to quantum, bosonic indistinguishability.

A first set of experiments aimed at measuring the bunching probabilities $p_b^{(q)}$ and  $p_b^{(c)}$ respectively of quantum (i.e. indistinguishable) and classical (i.e. distinguishable) photons after each interferometer. We note that these probabilities depend both on the interferometer's design and the input state used. A bunching event involves, by definition, the overlap of at least two photons in a single output mode. The classical bunching probability $p_b^{(c)}$ is obtained from single-photon experiments that characterize the transition probabilities between each input/output combination. To measure $p_b^{(q)}$, we set up experiments with $n$ input photons (each entering a different mode), and detected rates of $n$-fold coincidences of photons coming out in $n$ different modes of each chip.  As in a standard Hong-Ou-Mandel measurement, each experimental run was done in identical conditions and for the same time interval, varying only the delays that make the particles distinguishable or not, and so give us an estimate of the ratio $t\equiv(1-p_b^{(q)})/(1-p_b^{(c)})$. Together with our measured $p_b^{(c)}$, this allowed us to estimate the bunching probability for indistinguishable photons $p_b^{(q)}=1-t(1-p_b^{(c)})$.

We summarize the experimental results for a number of different photonic chips in Figs. \ref{twophotons} {\bf a-g} (two-photon experiments) and Fig. \ref{threephotons} (three-photon experiments). The results are in good agreement with theory, taking into account the partial indistinguishability of the photon source \cite{Tsuj04}. The shaded regions indicate the average bunching behavior obtained numerically from 10000 unitaries sampled from the uniform, Haar distribution. For all the employed interferometers we find that indistinguishable photons display a higher coincidence rate than distinguishable photons do ($p_b^{(q)}>p_b^{(c)}$); this is known to be true for averages \cite{Arkhipov11}. Furthermore, $p_b^{(q)}$ falls as $m$ increases, as predicted in \cite{Aaronson10, Arkhipov11}. This latter result is somewhat counter-intuitive, given the bunching behavior of bosons, and has been referred to as the bosonic birthday paradox \cite{Aaronson10}; in fact, it was shown that both $p_b^{(q)}$ and the bunching probability associated with a classical, uniform distribution decay with the same asymptotic behavior as $m$ increases \cite{Arkhipov11}.

\begin{figure}[ht!]
\includegraphics[width=8.5cm]{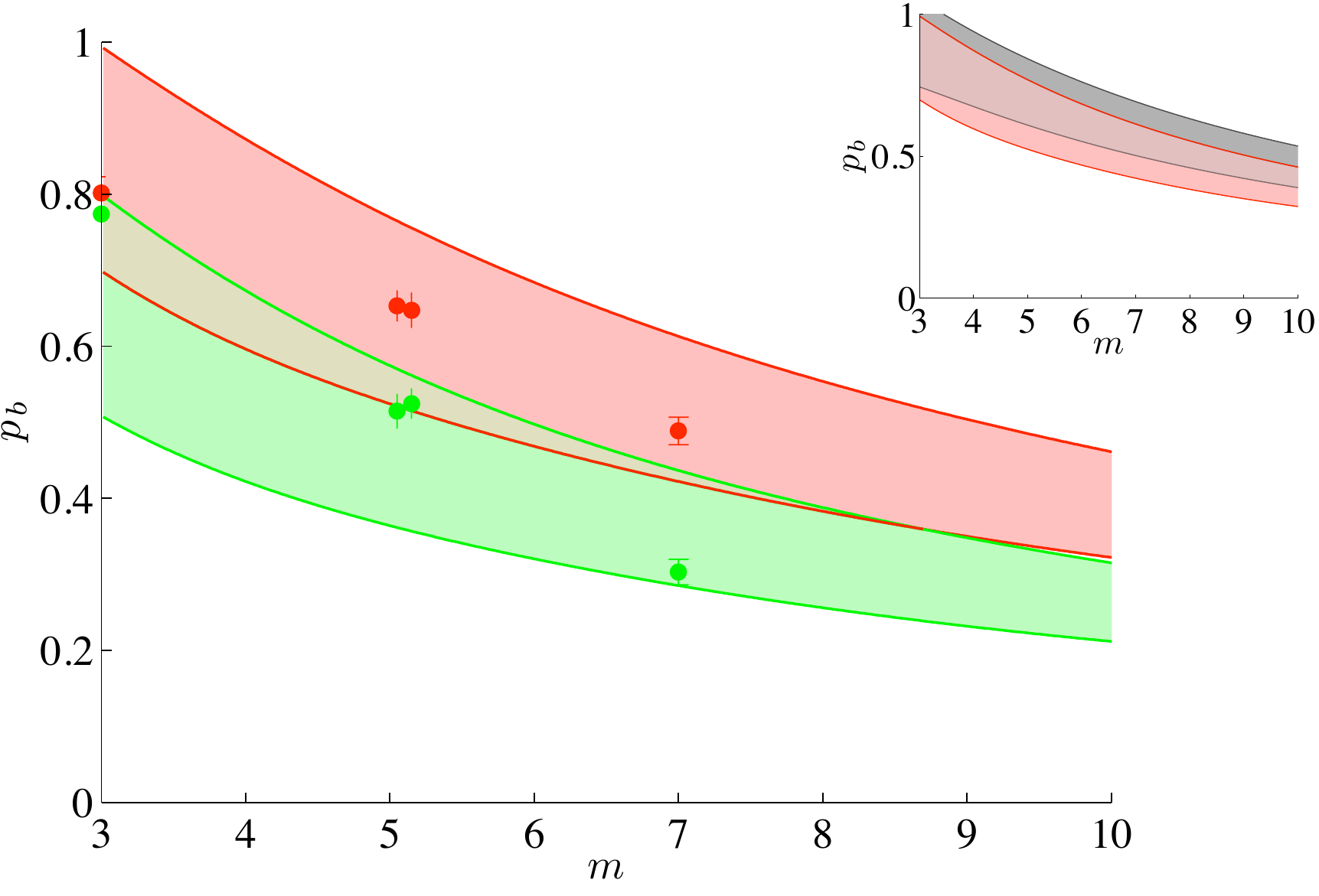}
\caption{{\textbf{Three-photon photonic bunching data.}} Experimental results for the three-photon photonic bunching experiments as a function of the number of modes $m$ for three indistinguishable photons ($p_{b}^{(q)}$, red points) and three distinguishable photons ($p_{b}^{(c)}$, green points). Shaded areas correspond to the interval $[\overline{p_{b}} - 1.5 \, \sigma; \overline{p_{b}} + 1.5 \sigma]$ with a numerical sampling over $10000$ uniformly random unitaries, $\overline{p_{b}}$ being the average bunching probability and $\sigma$ its standard deviation. Red area: simulation taking into account the partial indistinguishability parameter $\alpha$ of our source. Green area: three distinguishable photons. Error bars in the experimental data are due to the Poissonian statistics of the measured events, and where not visible are smaller than the symbol. Inset: numerical simulation of the effect of photon distinguishability on the bunching probability $p_{b}$. Grey area: perfectly indistinguishable photons ($\alpha=1$).}
\label{threephotons}
\end{figure}

It is interesting to compare this photonic bunching behavior with what is expected from fermions, since the Pauli exclusion principle forbids fermionic bunching. Two-particle fermionic statistics may be simulated by exploiting the symmetry of two-photon wave-functions in an additional degree of freedom \cite{Matt2011,Sans2012}. For this purpose, we injected the interferometers with two photons, in an anti-symmetric polarization-entangled state, in two different input ports \cite{Sans2012,Crespi2013}. The results are shown in Fig. \ref{twophotons} {\bf h}, where a suppression of the bunching probability can be observed for the case of simulated indistinguishable fermions. 

\begin{figure}[ht!]
\includegraphics[width=0.49 \textwidth]{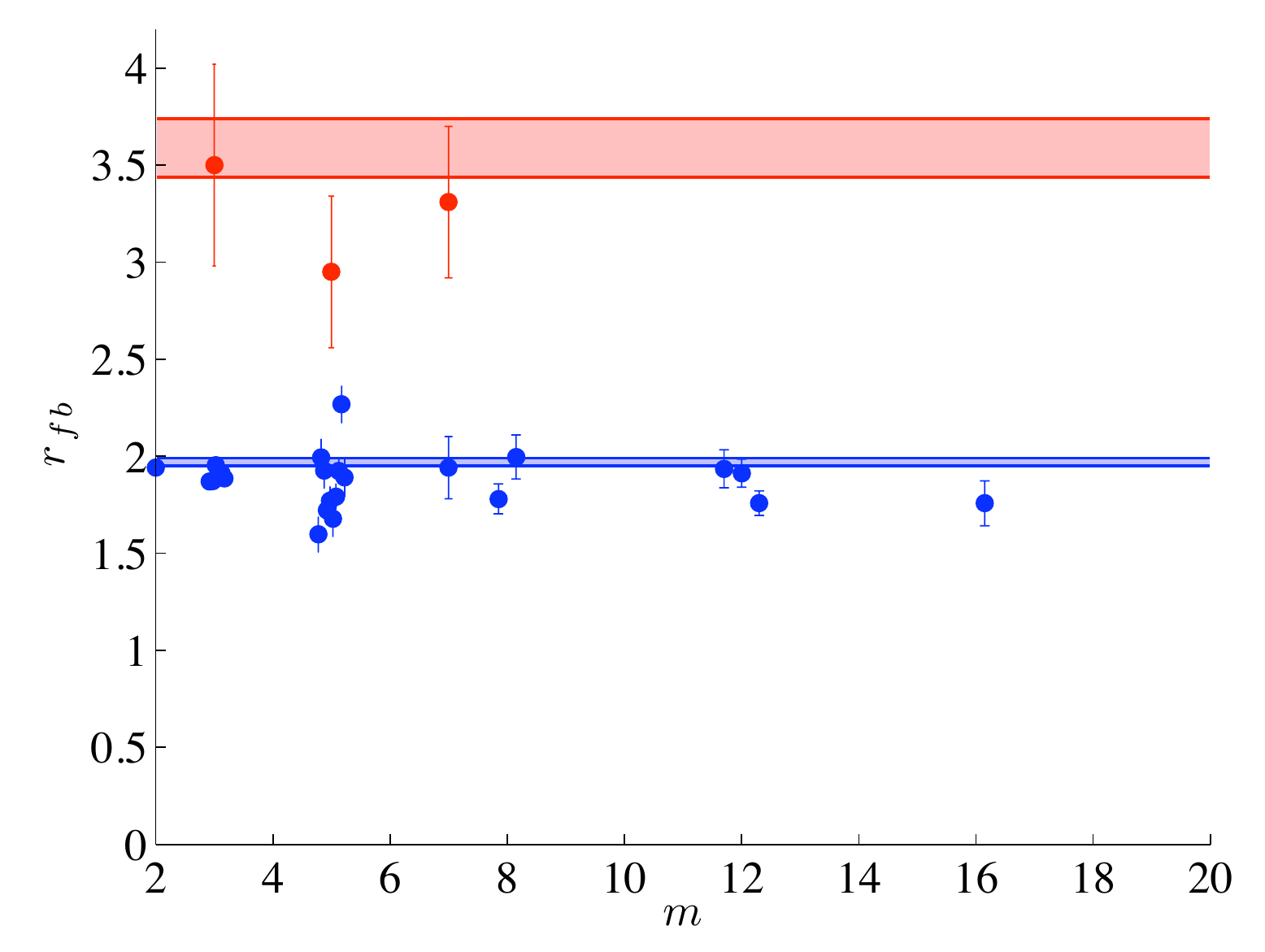}
\caption{\textbf{Ratio $r_{fb}$ between quantum and classical full-bunching probabilities.} Here we report $r_{fb}$ for two- and three-photon experiments on a number of photonic chips. Blue data: $r_{fb}$ for two-photon experiments in chips with different number of modes $m$. Blue area: expected value for two-photon full-bunching ratio. Two-photon measurements have been carried out in two different regimes, where the measured photon indistinguishability was $\beta \simeq 0.95$ (for the interferometers with $m=3,5,7$) and $\beta \simeq 0.99$ (for the interferometers with $m=2,8,12,16$). The corresponding values of the expected two-photon bunching ratio are given by $r_{fb}=\beta 2!+(1-\beta)(2-1)!$, and ranges between $r_{fb}\simeq 1.95$ and $r_{fb}\simeq 1.99$. Red data: $r_{fb}$ for three-photon experiments in various random chips. Red area: expected value for three-photon full-bunching ratio taking into account the partial indistinguishability parameter $\alpha$. Error bars in the experimental data are due to the Poissonian statistics of the measured events, and where not visible are smaller than the symbol.}
\label{nfactorial}
\end{figure}

We now turn to experiments that test our bosonic full-bunching rule. We estimated the quantum/classical full-bunching probability ratio $r_{fb}$ by introducing delays to change the distinguishability regime, and performing photon counting measurements in selected output ports, using fiber beam-splitters and multiple single-photon detectors. In Fig. \ref{nfactorial} (blue data) we plot the \mbox{(full-)}bunching ratio for all two-photon experiments referred to in Fig. \ref{twophotons}, and find good agreement with the predicted quantum enhancement factor of $2!=2$. Note that in two-photon experiments every bunching event is also a full bunching event, which means that when $n=2$ the ratio $r_{fb}=r_b=2$, independently of the number of modes $m$.

We have also measured  three-photon, full-bunching probabilities in random interferometers with number of modes $m=3,5,7$. Perfectly indistinguishable photons would result in the predicted $3!=6$-fold quantum enhancement for full-bunching probabilities. The partial indistinguishability  $\alpha=0.63 \pm 0.03$ of our three injected photons reduces this quantum enhancement to a factor $r_{fb}=\alpha^2 \, 3! + (1-\alpha^2) \, (3-1)!=3.59 \pm 0.15$.  The results can be seen in Fig. \ref{nfactorial} (red data), showing good agreement with the predicted value.

In conclusion, our experiments characterize the bunching behavior of up to three photons evolving in a variety of integrated multimode circuits. Our results are in agreement with the recent predictions of \cite{Aaronson10, Arkhipov11}, regarding the average bunching behavior of bosons in random interferometers. We have also proved a new rule that sharply discriminates quantum and classical behavior, by focusing on events in which all photons exit the interferometer bunched in a single mode. We have obtained experimental confirmation also of this new full-bunching law. Besides its fundamental importance in the description of bosonic quantum systems, the bunching behavior of bosons we studied here can be exploited in contexts ranging from quantum computation to quantum metrology \cite{Lucke2011}.

\textbf{Acknowledgements.} This work was supported by the ERC-Starting Grant 3D-QUEST (3D-Quantum Integrated Optical Simulation; grant agreement no. 307783): http://www.3dquest.eu. D.B. and E.G. acknowledge support from the Brazilian National Institute for Science and Technology of Quantum Information (INCT-IQ/CNPq). We acknowledge support from Giorgio Milani and Sandro Giacomini.


%

\end{document}


\title{General rules for bosonic bunching in multimode interferometers\\
- Supplementary Material -}

\author{Nicol\`o Spagnolo}
\affiliation{Dipartimento di Fisica, Sapienza Universit\`{a} di Roma,
Piazzale Aldo Moro 5, I-00185 Roma, Italy}

\author{Chiara Vitelli}
\affiliation{Center of Life NanoScience @ La Sapienza, Istituto
Italiano di Tecnologia, Viale Regina Elena, 255, I-00185 Roma, Italy}
\affiliation{Dipartimento di Fisica, Sapienza Universit\`{a} di Roma,
Piazzale Aldo Moro 5, I-00185 Roma, Italy}

\author{Linda Sansoni}
\affiliation{Dipartimento di Fisica, Sapienza Universit\`{a} di Roma,
Piazzale Aldo Moro 5, I-00185 Roma, Italy}

\author{Enrico Maiorino}
\affiliation{Dipartimento di Fisica, Sapienza Universit\`{a} di Roma,
Piazzale Aldo Moro 5, I-00185 Roma, Italy}

\author{Paolo Mataloni}
\affiliation{Dipartimento di Fisica, Sapienza Universit\`{a} di Roma,
Piazzale Aldo Moro 5, I-00185 Roma, Italy}
\affiliation{Istituto Nazionale di Ottica (INO-CNR), Largo E. Fermi 6, I-50125 Firenze, Italy}

\author{Fabio Sciarrino}
\affiliation{Dipartimento di Fisica, Sapienza Universit\`{a} di Roma,
Piazzale Aldo Moro 5, I-00185 Roma, Italy}
\affiliation{Istituto Nazionale di Ottica (INO-CNR), Largo E. Fermi 6, I-50125 Firenze, Italy}

\author{Daniel J. Brod}
\affiliation{Instituto de F\'isica, Universidade Federal Fluminense, Av. Gal. Milton Tavares de Souza s/n, Niter\'oi, RJ, 24210-340, Brazil }

\author{Ernesto F. Galv\~{a}o}
\affiliation{Instituto de F\'isica, Universidade Federal Fluminense, Av. Gal. Milton Tavares de Souza s/n, Niter\'oi, RJ, 24210-340, Brazil }

\author{Andrea Crespi}
\affiliation{Istituto di Fotonica e Nanotecnologie, Consiglio
Nazionale delle Ricerche (IFN-CNR), Piazza Leonardo da Vinci, 32,
I-20133 Milano, Italy}
\affiliation{Dipartimento di Fisica, Politecnico di Milano, Piazza
Leonardo da Vinci, 32, I-20133 Milano, Italy}

\author{Roberta Ramponi}
\affiliation{Istituto di Fotonica e Nanotecnologie, Consiglio
Nazionale delle Ricerche (IFN-CNR), Piazza Leonardo da Vinci, 32,
I-20133 Milano, Italy}
\affiliation{Dipartimento di Fisica, Politecnico di Milano, Piazza
Leonardo da Vinci, 32, I-20133 Milano, Italy}

\author{Roberto Osellame}
\affiliation{Istituto di Fotonica e Nanotecnologie, Consiglio
Nazionale delle Ricerche (IFN-CNR), Piazza Leonardo da Vinci, 32,
I-20133 Milano, Italy}
\affiliation{Dipartimento di Fisica, Politecnico di Milano, Piazza
Leonardo da Vinci, 32, I-20133 Milano, Italy}

\maketitle

\begin{figure}[ht!]
\centering
\renewcommand{\figurename}{{\bf Supplementary Figure}}
\includegraphics[width=0.66\textwidth]{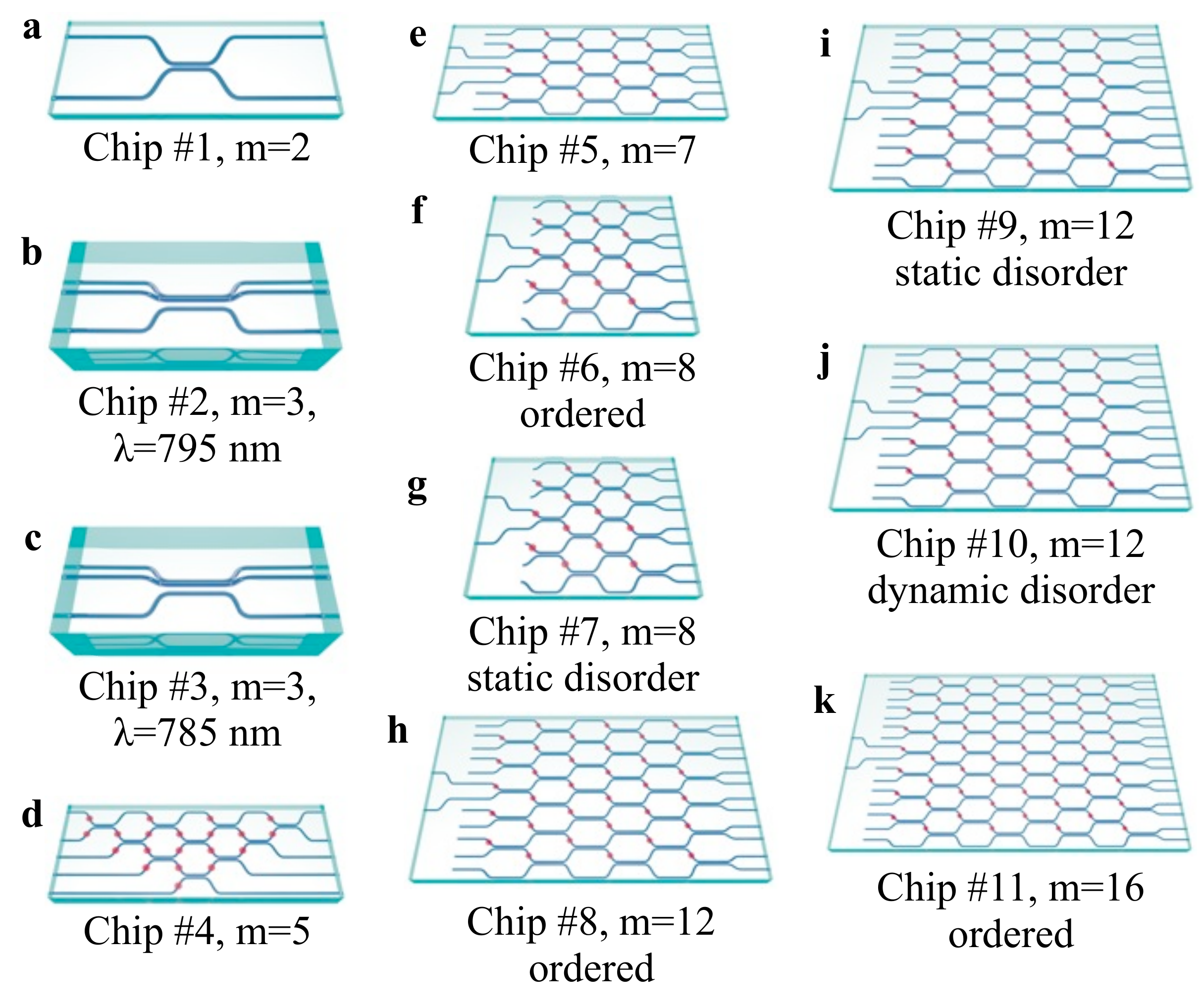}
\caption{{\bf Integrated interferometers architectures}. Schemes of the internal structures of the tested integrated devices. Red spots represent phase shifters. {\bf a}, Two-mode symmetric directional coupler (see Ref. \cite{Sans2010}). {\bf b}, Three-dimensional tritter tested with photons at $795$ nm (see Ref. \cite{Spag2012}). {\bf c}, Three-dimensional tritter tested with photons at $785$ nm (see Ref. \cite{Spag2012}). {\bf d}, Arbitrary random $5$-mode interferometer (see Ref. \cite{Crespi2012}). {\bf e}, Random $7$-mode, $5$-layer interferometer with arbitrary phases and symmetric ($T=0.5$) directional couplers. {\bf f}, $8$-mode, $4$-layer quantum walk with equal phases and symmetric directional couplers (see Ref. \cite{Sans2012}). {\bf g}, $8$-mode, $4$-layer interferometer with symmetric directional couplers and static disorder encoded in the phase pattern (shown in Ref. \cite{Crespi2013}). {\bf h}, $12$-mode, $6$-layer interferometer with equal phases and symmetric directional couplers (see Ref. \cite{Crespi2013}). {\bf i}, $12$-mode, $6$-layer interferometer with symmetric directional couplers and static disorder encoded in the phase pattern (shown in Ref. \cite{Crespi2013}). {\bf j}, $12$-mode, $6$-layer interferometer with symmetric directional couplers and dynamic disorder encoded in the phase pattern (shown in Ref. \cite{Crespi2013}). {\bf k}, $16$-mode, $8$-layer quantum walk with equal phases and symmetric directional couplers (see Ref. \cite{Crespi2013}).}
\label{fig:S1}
\end{figure}

\newpage

\begin{table}[ht!]
\renewcommand{\tablename}{{\bf Supplementary Table}}
\footnotesize
\begin{tabular}{|c|c|c||c|c||c|c||c|c||c|}
\hline
m & Chip $\#$ &Input modes & $p^{(q)}_{b}$, $U_{t}$ ($U_{r}$) & $p^{(q)}_{b}$, exp & $p^{(c)}_{b}$, $U_{t}$ ($U_{r}$) & $p^{(c)}_{b}$, exp & $r_{fb}$, $U_{t}$ ($U_{r}$) & $r_{fb}$, exp &  Laser mode\\
\hline
\hline
2 & 1 & (1,2) & 1 & $0.965 \pm 0.003$& 0.5 & $0.497 \pm 0.031$& 2 & $1.943 \pm 0.004$ & CW\\
\hline
3 & 2 & (1,2) & 0.667 ($0.6636 \pm 0.0006$) & $0.625 \pm 0.005$ & 0.333 ($0.3319 \pm 0.0003$) & $0.3304 \pm 0.0008$ & 2 & $1.89 \pm 0.02$ & PW\\
\hline
3 & 2 & (1,3) & 0.667 ($0.6656 \pm 0.0004$) & $0.635 \pm 0.004$ & 0.333 ($0.3330 \pm 0.0002$) & $0.3319 \pm 0.0008$ & 2 & $1.91 \pm 0.02$ & PW\\
\hline
3 & 2 & (2,3) & 0.667 ($0.6661 \pm 0.0003$) & $0.629 \pm 0.006$ & 0.333 ($0.3332 \pm 0.0002$) & $0.3335 \pm 0.0008$ & 2 & $1.89 \pm 0.02$ & PW\\
\hline
3 & 3 & (1,2) & 0.667 ($0.652 \pm 0.001$) & $0.610 \pm 0.003$ & 0.333 ($0.3261 \pm 0.0006$) & $0.326 \pm 0.002$ & 2 & $1.87 \pm 0.02$ & PW\\
\hline
3 & 3 & (1,3) & 0.667 ($0.652 \pm 0.001$) & $0.613 \pm 0.002$ & 0.333 ($0.3262 \pm 0.0006$) & $0.327 \pm 0.002$ & 2 & $1.87 \pm 0.02$ & PW\\
\hline
3 & 3 & (2,3) & 0.667 ($0.6718 \pm 0.0005$) & $0.654 \pm 0.003$ & 0.333 ($0.3360 \pm 0.0003$) & $0.335 \pm 0.002$ & 2 & $1.95 \pm 0.02$ & PW\\
\hline
5 & 4 & (1,2) & 0.542 ($0.375 \pm 0.005$) & $0.370 \pm 0.019$ & 0.271 ($0.188 \pm 0.002$) & $0.232 \pm 0.005$ & 2 & $1.60 \pm 0.09$ & PW\\
\hline
5 & 4 & (1,3) & 0.269 ($0.411 \pm 0.005$) & $0.325 \pm 0.014$ & 0.134 ($0.206 \pm 0.002$) & $0.163 \pm 0.003$ & 2 & $1.99 \pm 0.09$ & PW\\
\hline
5 & 4 & (1,4) & 0.206 ($0.296 \pm 0.004$) & $0.248 \pm 0.011$ & 0.103 ($0.148 \pm 0.002$) & $0.129 \pm 0.003$ & 2 & $1.93 \pm 0.09$ & PW\\
\hline
5 & 4 & (1,5) & 0.188 ($0.251 \pm 0.004$) & $0.359 \pm 0.013$ & 0.094 ($0.125 \pm 0.002$) & $0.208 \pm 0.004$ & 2 & $1.72 \pm 0.07$ & PW\\
\hline
5 & 4 & (2,3) & 0.388 ($0.350 \pm 0.005$) & $0.354 \pm 0.013$ & 0.194 ($0.175 \pm 0.002$) & $0.200 \pm 0.003$ & 2 & $1.77 \pm 0.07$ & PW\\
\hline
5 & 4 & (2,4) & 0.304 ($0.390 \pm 0.004$) & $0.278 \pm 0.014$ & 0.152 ($0.195 \pm 0.002$) & $0.166 \pm 0.003$ & 2 & $1.68 \pm 0.09$ & PW\\
\hline
5 & 4 & (2,5) & 0.303 ($0.373 \pm 0.004$) & $0.349 \pm 0.011$ & 0.152 ($0.186 \pm 0.002$) & $0.194 \pm 0.004$ & 2 & $1.79 \pm 0.07$ & PW\\
\hline
5 & 4 & (3,4) & 0.122 ($0.280 \pm 0.004$) & $0.474 \pm 0.012$ & 0.061 ($0.140 \pm 0.002$) & $0.246 \pm 0.005$ & 2 & $1.92 \pm 0.06$ & PW\\
\hline
5 & 4 & (3,5) & 0.239 ($0.293 \pm 0.005$) & $0.321 \pm 0.011$ & 0.112 ($0.146 \pm 0.003$) & $0.141 \pm 0.002$ & 2 & $2.27 \pm 0.09$ & PW\\
\hline
5 & 4 & (4,5) & 0.653 ($0.517 \pm 0.003$) & $0.241 \pm 0.012$ & 0.326 ($0.259 \pm 0.002$) & $0.128 \pm 0.002$ & 2 & $1.89 \pm 0.10$ & PW\\
\hline
7 & 5 & (3,4) & 0.386 & $0.325 \pm 0.019$ & 0.193 & $0.167 \pm 0.010$ & 2 & $1.94 \pm 0.16$ & PW\\
\hline
8 & 6 & (4,5) & $0.212$ & $0.217 \pm 0.006$ & $0.106$ & $0.122 \pm 0.004$ & 2 & $1.78 \pm 0.08$ & CW\\
\hline
8 & 7 & (4,5) & $0.188$ & $0.149 \pm 0.006$ & $0.094$ & $0.075 \pm 0.004$ & 2 & $2.00 \pm 0.11$ & CW\\
\hline
12 & 8 & (6,7) & $0.298$ & $0.263 \pm 0.012$ & $0.149$ & $0.136 \pm 0.003$ & 2 & $1.93 \pm 0.10$ & CW\\
\hline
12 & 9 & (6,7) & $0.262$ & $0.220 \pm 0.006$ & $0.131$ & $0.115 \pm 0.003$ & 2 & $1.91 \pm 0.07$ & CW\\
\hline
12 & 10 & (6,7) & $0.240$ & $0.223 \pm 0.006$ & $0.120$ & $0.127 \pm 0.003$ & 2 & $1.76 \pm 0.06$ & CW\\
\hline
16 & 11 & (8,9) & $0.113$ & $0.133 \pm 0.007$ & $0.056$ & $0.076 \pm 0.003$ & 2 & $1.76 \pm 0.12$ & CW\\
\hline
\end{tabular}
\caption{{\bf Experimental two-photon data}. Comparison between the theoretical predictions and the experimental results for the bunching probabilities ($p^{(q)}_{b}$ for indistinguishable photons, $p^{(c)}_{b}$ for distinguishable particles) and for the bunching ratios ($r_{fb}=p^{(q)}_{b}/p^{(c)}_{b}$) of two-boson experiments. Additionally, for chips (2,3,4) we report in brackets the predictions obtained with the tomographically reconstructed matrices $U_{r}$ of the devices (see Refs. \cite{Spag2012,Crespi2012} for details). Error bars in the predictions with the reconstructed matrices are due to the reconstruction errors and have been obtained by a Monte-Carlo simulation with $M=1000$ sampled unitaries. The data have been collected in two different regimes, corresponding to the two-photon source working with a continuous wave (CW, $\lambda = 810$ nm and $806$ nm) or a pulsed wave (PW, $\lambda = 785$ nm and $\lambda = 795$ nm) pump laser. In the pulsed regime, the expected two-photon visibilities are reduced by a factor $q \simeq 0.95$ with respect to the calculated value, where $q$ takes into account the imperfect indistinguishability of the two photons. This correction on the expected visibilities is not reported in the table.}
\label{tab:S1}
\end{table}

\begin{table}[ht!]
\renewcommand{\tablename}{{\bf Supplementary Table}}
\footnotesize
\begin{tabular}{|c|c|c||c|c||c|c||}
\hline
m & Chip $\#$ &Input modes & $p^{(q)}_{f}$, $U_{t}$ & $p^{(q)}_{f}$, exp & $p^{(c)}_{f}$, $U_{t}$ & $p^{(c)}_{f}$, exp \\
\hline
\hline
2 & 1 & (1,2) & 0 & $0.030 \pm 0.061$ & 0.5 & $0.497 \pm 0.031$ \\
\hline
8 & 6 & (4,5) & 0 & $0.062 \pm 0.003$ & 0.106 & $0.122 \pm 0.004$ \\
\hline
8 & 7 & (4,5) & 0 & $0.039 \pm 0.002$ & 0.094 & $0.075 \pm 0.003$ \\
\hline
12 & 8 & (6,7) & 0 & $0.041 \pm 0.003$ & 0.149 & $0.136 \pm 0.003$ \\
\hline
12 & 9 & (6,7) & 0 & $0.045 \pm 0.002$ & 0.131 & $0.115 \pm 0.003$ \\
\hline
16 & 11 & (8,9) & 0 & $0.031 \pm 0.003$ & 0.056 & $0.076 \pm 0.003$ \\
\hline
\end{tabular}
\caption{{\bf Experimental two-fermion data}. Comparison between the theoretical predictions and the experimental results for the bunching probabilities ($p^{(q)}_{f}$ for indistinguishable fermions, $p^{(c)}_{f}$ for distinguishable particles) of two-fermion experiments. All data in this case have been collected with a continuous wave (CW) source. The fermionic behavior is mimicked by injecting the integrated interferometers with a two-photon antisymmetric polarization-entangled state $\vert \psi^{-} \rangle = 2^{-1/2} (\vert H \rangle_{i} \vert V \rangle_{j} - \vert V \rangle_{i} \vert H \rangle_{j})$, where $(i,j)$ are the input modes. The observed non-zero bunching probability for two indistinguishable fermions $p^{(q)}_{f}$ is due to the imperfect preparation of the polarization state $\vert \psi^{-} \rangle$.}
\label{tab:S2}
\end{table}

\newpage

\begin{table}[ht!]
\renewcommand{\tablename}{{\bf Supplementary Table}}
\footnotesize
\begin{tabular}{|c|c|c||c|c||c|c||c|c|}
\hline
m & Chip $\#$ &Input modes & $p^{(q)}_{b}$, $U_{t}$ ($U_{r}$) & $p^{(q)}_{b}$, exp & $p^{(c)}_{b}$, $U_{t}$ ($U_{r}$) & $p^{(c)}_{b}$, exp & $r_{fb}$, $U_{t}$ ($U_{r}$) & $r_{fb}$, exp\\
\hline
\hline
3 & 3 & (1,2,3) & 0.795 ($0.8008 \pm 0.0004$) & $0.801 \pm 0.021$ & 0.778 ($0.7734 \pm 0.0004$) & $0.774 \pm 0.001$ & 3.588 & $3.50 \pm 0.52$\\
\hline
5 & 4 & (1,2,3) & 0.694 ($0.675 \pm 0.002$) & $0.653 \pm 0.020$ & 0.469 ($0.483 \pm 0.002$) & $0.515 \pm 0.023$ & 3.588 & N/A \\
\hline
5 & 4 & (2,3,4) & 0.654 ($0.635 \pm 0.002$) & $0.647 \pm 0.023$ & 0.552 ($0.535 \pm 0.002$) & $0.525 \pm 0.019$ & 3.588 & $2.89 \pm 0.39$\\
\hline
7 & 5 & (3,4,5) & 0.537 & $0.489 \pm 0.018$ & 0.331 & $0.303 \pm 0.017$& 3.588 & $3.31 \pm 0.39$\\
\hline
\end{tabular}
\caption{{\bf Experimental three-photon data}. Comparison between the theoretical predictions and the experimental results for the bunching probabilities ($p^{(q)}_{b}$ for indistinguishable photons, $p^{(c)}_{b}$ for distinguishable particles) and for the full-bunching ratios ($r_{fb}=p^{(q)}_{fb}/p^{(c)}_{fb}$) of three-boson experiments. Additionally, for the chips (3,4) we report in brackets the predictions obtained with the tomographically reconstructed matrices $U_{r}$ of the devices (see Refs. \cite{Spag2012,Crespi2012} for details). Error bars in the predictions with the reconstructed matrices are due to the reconstruction errors and have been obtained by a Monte-Carlo simulation with $M=1000$ sampled unitaries. All data in this case have been collected with a pulsed wave (PW) source with one partially distinguishable photon (see Ref. \cite{Spag2012,Crespi2012} for more details). Due to the indistinguishability factor $\alpha=0.63 \pm 0.03$, the full-bunching ratio $r_{fb}$ for all the $n=3$ photons coming out from the same port is a mixture of the $n!$ term (three indistinguishable photons) and the $(n-1)!$ one (two indistinguishable photon and one distinguishable one), and equals $r_{fb}=\alpha^2 \, n! + (1-\alpha^2) \, (n-1)!=3.59 \pm 0.15$.}
\label{tab:S3}
\end{table}

\section{Experimental details}
Two different apparata have been exploited in order to test different chip behaviours.

\textbf{1- Apparatus with cw source-} This setup has been adopted for testing the chips reported in table S1 (labelled by CW) and S2. 

Chips number $6-9-10$ have been tested by a  polarization entangled photon pairs source working at wavelength $806 nm$. The entangled photons were generated via spontaneous parametric down conversion in a $1.5 mm$  $\beta-$barium borate crystal (BBO) cut for type-II noncollinear phase matching, pumped by a cw diode laser with power $50 mW$. The complete experimental setup involved a  $\beta$-barium borate crystal (BBO) cut for type II phase matching and pumped by a CW diode laser ($\lambda=403\ nm$), which generated polarization-entangled photon pairs in the state $\ket{\psi^-}_{AB}=\frac{1}{\sqrt{2}}(\ket{H}_A\ket{V}_B-\ket{V}_A\ket{H} _B)$ at $\lambda=806\ nm$ via spontaneous parametric down conversion .  Starting from the singlet state $\ket{\psi^-}$, the other three Bell states are obtained by suitable rotation of waveplates  inserted on modes $A$ and $B$. This allowed to simulate fermionic and bosonic behaviour. Photon pairs were then delivered into the  circuit  through two single mode fibers. Standard polarization controllers  and a voltage-controlled liquid crystal device  were adopted to compensate polarization rotations induced by propagation of the beams along the fibers.
The output of the chip was collected and collimated by a first lens, divided by a bulk beam splitter. The different outputs were coupled  by two lenses  (one per branch), to multimode fibers  mounted on motorized stages. By translating the stages with respect to the beams, different outputs could be coupled into the multi mode fibers. Photons were finally detected by single photon counting modules and coincidence counts were recorded.

Measurements corresponding to the chips $1-7-8-11$ with two photons and for the fermionic case were performed analogously by a BBO parametric down conversion source at 810 nm in the cw regime. Two-photon Hong-Ou-Mandel interference in this case is $V_{(2)} \simeq 0.99$.

\textbf{2- Apparatus with pulsed source-}This setup has been adopted for testing the chips reported in table S1 (labelled by PW) and S3.

Chips $3-4-5$ have been tested by a  four photons source shined by a mode locked Coherent Mira laser ($785$  nm), with an average power of  $1.6 W$, and pumped by a cw  Coherent Verdi V18 laser, working at $12 W$  output power. The output beam was doubled in frequency by second harmonic generation ( $392.5$ nm) leading to a pump beam with $750 mW$ average power.

Chip 2 has been tested by adopting the same laser source tuned at $795nm$. 

In both cases photons were produced in the pulsed regime at $785 nm$ ($795nm$) exploiting the parametric down conversion process by pumping a 2 mm long BBO crystal by the $392.5 nm$ ($397.5 nm$) wavelength pump field. Typical count rates of the source were around $250000 Hz$ for the four signals, $40000 Hz$ for the two-fold coincidences and $20 Hz$ for the fourfold coincidences. Spectral filtering by $3 nm$ interferential filters, coupling into single-mode fibers and propagation through different delay lines were performed before coupling into the chips. Light passing from single mode fibers and the chip was performed through a fiber array (coupling efficiency mearly $0.7$), typical coupling between the fiber array and the chip was around $0.7$, and the chip trasmissivity ranged between $0.3-0.6$ depending of the internal  geometry. Delay lines were exploited before the fiber array to modulate the different interference regimes, corresponding to the cases of distinguishable and indistinguishable particles. 
 The measured visibility for two-photon Hong-Ou-Mandel interference in a symmetric 50/50 beam-splitter between photons belonging to the same pair was $V_{(2)}=0.952 \pm 0.006$. Hong-Ou-Mandel interference in a symmetric beam-splitter between photons belonging to different pairs was $V^{\prime}_{(2)}=0.63 \pm 0.03$.
With three-photon measurements, one of the four generated photons was adopted as a trigger for coincidence detection, while the other three photons were coupled inside the chip. The output modes were detected by using multimode fibers and single-photon avalanche photodiodes. Typical coupling efficiency into the multi-mode fibers was around $0.8$.  Coincidences between different detectors were used to reconstruct the probability of obtaining a given output state. The full-bunching contributions were measured by splitting in two (three) equal parts the desired output mode with fiber beam-splitters and by measuring two-fold (three-fold) coincidences.